%% file: lsh_matching.tex
\documentclass[11pt]{article}

\usepackage{fullpage}
\usepackage{amsmath,amsfonts,amssymb,amsthm}
\usepackage{mathtools}
\usepackage{graphicx}
\PassOptionsToPackage{hyphens}{url}\usepackage{hyperref}
\usepackage{algorithm2e}

\theoremstyle{plain}
\theoremstyle{plain}
\newtheorem{thm}{Theorem}[]
\newtheorem{defn}{Definition}[]

\newcommand{\HF}{\mathcal{H}}
\newcommand{\HG}{\mathcal{G}}

\title{When Hashing Met Matching: \\ Efficient Spatio-Temporal Search for Ridesharing}

\author{Chinmoy Dutta}
\date{}

\begin{document}
\maketitle

\input{abstract.tex}

\pagebreak

\input{intro.tex}
\input{related_literature.tex}
\input{prelims.tex}

\input{algorithm.tex}
\input{experiment.tex}
\input{extensions.tex}

\bibliographystyle{plain}
\bibliography{lsh_matching}


\end{document}

%% file: abstract.tex
\begin{abstract}
  Carpooling, or sharing a ride with other passengers, holds immense potential for urban transportation. Ridesharing platforms enable such sharing of rides using real-time data. Finding ride matches in real-time at urban scale is a difficult combinatorial optimization task and mostly heuristic approaches are applied. In this work, we mathematically model the problem as that of finding near-neighbors and devise a novel efficient spatio-temporal search algorithm based on the theory of locality sensitive hashing for Maximum Inner Product Search (MIPS). The proposed algorithm can find $k$ near-optimal potential matches for every ride from a pool of $n$ rides in time $O(n^{1 + \rho} (k + \log n) \log k)$ and space $O(n^{1 + \rho} \log k)$ for a small $\rho < 1$. Our algorithm can be extended in several useful and interesting ways increasing its practical appeal. Experiments with large NY yellow taxi trip datasets show that our algorithm consistently outperforms state-of-the-art heuristic methods thereby proving its practical applicability.
\end{abstract}

%% file: intro.tex
\section{Introduction}
\label{sec:intro}
In this paper, we consider carpooling or sharing a ride with other passengers. The importance of this problem for planning urban transportation and designing "smart cities" lies in its promise to provide a solution for serious urban issues of excessive traffic congestion, resource consumption and air pollution~\cite{AS94,CAWYB13}. 



Real-time information and monitoring of urban mobility and the ability to do large scale computation on real-time data together allow ridesharing platforms to enable sharing of rides to unprecedented levels by exploiting unused vehicle capacity. A number of works~\cite{S14,AESW11,MZW13} have studied urban traffic and confirmed the tremendous potential of real-time urban-scale ridesharing to reduce the burden on urban transportation. In order to balance the costs (detours, increased travel time, loss of privacy) with the utility (traffic decongestion, decreased resource consumption and air pollution, cheaper rides, use of high-occupancy lanes, marketplace efficiency) of ridesharing, matching rides in a way that minimizes cost and maximizes utility of sharing is of paramount importance for the success of ridesharing. 
 
The graph-theoretic framework of {\em shareability networks}, introduced in~\cite{S14} and extended in the influential work of Alonso-Mora et al.~\cite{ASWFR17}, has been the most promising approach to tackle the challenge of urban-scale, real-time ride matching. The nodes of the network are either rides or driver routes. An edge between two rides represents feasibility of matching them together while an edge between a ride and and a driver route represents feasibility of adding the ride to the driver route. Here, feasibility means that the platform-defined user experience constraints (like maximum pickup wait, maximum detour etc.) are satisfied for all the rides concerned. The basic building block of shareability network has been used in several research works~\cite{ASWFR17,SMG19,S14}, which build on top of this framework and solve various matching problems (bipartite, non-bipartite, hypergraph) to enable real-time ridesharing.



Computing the real-time shareability network, however, is a computationally intensive task. Comparing nodes with each other for feasibility has a time complexity of $O(n^2)$ and requires $O(n^2)$ calls to a routing service for a match pool with $n$ nodes (rides or driver routes). For a ride-sharing platform, it is common to have several thousand rides in the match pool for a large, densely populated urban region at peak commute times. Such platforms also employ features such as batching of arriving ride requests (for example, Lyft shared saver, Uber express pool), swapping already made matches etc. which further thicken the match pool. To generate the shareability network by a quadratic complexity search is therefore ruled out.

Heuristic-based methods have been proposed and comprise the state-of-the-art to deal with this combinatorial challenge. These methods output a small set of potential matches for each node using some heuristic search algorithm. For example, ~\cite{ASWFR17} suggested constructing edges only between ride requests with nearby pick-up locations. However, it is not uncommon to find large number of co-located ride requests at roughly the same time (e.g outside train stops when trains arrive, outside event venues when events end). Moreover, a service that lets passengers wait longer for cheaper rides or lets passengers schedule their rides in advance should be able to make matches between rides with far-away pickups. Another class of heuristics involves computing Haversine overlap between matches in the hope that it captures the true utility of matching. These methods need a way to contain the search space (brute force is still quadratic) and suffer from ignoring the real road network.

\subsection{Our contributions}
\label{subsec:contribution}
In this work, we propose a principled approach to tackle this combinatorial challenge. We devise an efficient randomized spatio-temporal search algorithm to construct a sparse, utility-aware shareability network.

In order to efficiently find a small set of high utility potential matches for each node, we use the theory of locality sensitive hashing (LSH). In particular, we use LSH construction for Maximum Inner Product Search (MIPS). At a high level, our approach works by finding suitable vector representations for rides and driver routes in a high-dimensional ambient vector space that capture their physical routes. Next, we define a similarity metric for this space that respects matching utility: larger similarity between vector representations imply higher matching utility between nodes. We then construct a locality sensitive hashing data structure for storing the vector representations that allows efficient search for similar vectors: for every vector, we can find $k$ (say $10$) approximately most-similar vectors (according to our similarity measure) efficiently in time $O(n^{1 + \rho} (k + \log n) \log k)$ and space $O(n^{1 + \rho} \log k)$ for a small $\rho < 1$. In the process, we only make $n$ calls to a routing service, one for each ride. This is sufficient to allow our algorithm to exploit the knowledge of the real road network.

In the rest of the paper, for simplicity of exposition, we only talk of matching rides, and describe how to handle driver routes in the last section. Our algorithm can be extended in several interesting ways. It works for matching utility defined in terms of any cost function which is linear in the route (e.g. travel duration, travel distance, tolls etc.). Moreover, it can handle time-varying cost functions (e.g. varying with traffic). It can work on a hybrid ridesharing system that provides real-time, on-demand rides as well as pre-scheduled ones. It even allows incremental computation. All these extensions are sketched in Section~\ref{sec:opt_ext}.

In order to demonstrate the practicality of our algorithm, we conducted large-scale experiments using publicly available NY yellow taxi trip datasets under different traffic patterns and varying load. Results showed that our algorithm consistently outperforms the best state-of-the-art heuristic method by as much as $6\%$.

To summarize, our contributions include:
\begin{enumerate}
    \item A novel principled approach for an efficient randomized construction of a sparse, real-time, utility-aware shareability network.
    
    \item Several interesting and useful extensions to the proposed algorithm that makes it even more practically appealing.
    
    \item Demonstration of practical utility of the algorithm via large-scale experiments.
\end{enumerate}

%% file: related_literature.tex
\section{Related literature}
\label{sec:lit_review}
The cost and effects of congestion for urban transportation has been widely studied~\cite{AS94,CAWYB13,PH13,CL12}.

Many of the studies related to mobility-on-demand and fleet management consider ridesharing without pooling rides, e.g~\cite{PSFR12,ZP14,SSGF16}. Heuristic-based solutions to matching problems were studied in~\cite{AESW11,MZW13}. There has also been a lot of research interest in studying autonomous ride-sharing systems~\cite{PSFR12,SSGF16,CA16}.

Recently, there has been a lot of research interest in carpooling and associated challenges. Quite a few studies have analyzed urban traffic and estimated huge potential for sharing rides~\cite{S14,AESW11,MZW13}. In fact, Santi et al~\cite{S14} showed about $80\%$ of rides in Manhattan can be shared by two riders.

Alonso-Mora et al.~\cite{ASWFR17} studied real-time high-capacity ride-shar-ing by extending the framework of shareability networks and proposed heuristic methods to construct the network based on spatial proximity of ride pickups. Simonetto et al.~\cite{SMG19} proposed a heuristic to speeden up matching at a loss of efficiency by matching arriving ride requests in small chunks by solving a bipartite matching at very high frequency. The greedy nature of their algorithm makes the system increasingly inefficient with increasing batch window of matching such as employed in Lyft shared saver, Uber express pool. Other approaches for sharing rides have been proposed in~\cite{BS15,STY15}. Very Recently, carpooling has also been formulated and studied in formal graph-theoretic online matching settings~\cite{HKTWZZ18,ABDJSS19}.

This nearest neighbor search (NNS) is a problem of major importance in several areas of science and engineering. Approximate nearest neighbor search techniques based on locality sensitive hashing (LSH) was introduced in~\cite{IM98,GIM99}. Owing to being parallelizable and suitable for high dimensional data, these techniques have found widespread use in research and industrial practice~\cite{GSM03,B01,JDS11,STSMIMD13}. However, LSH based techniques have not been introduced for the ride matching problem yet.

The approximate Maximum Inner Product search (MIPS) is a fundamental problem with variety of applications in areas such as recommendation systems~\cite{LCYM17,SRJ05}, deep learning~\cite{SS17} etc. The concept of Asymmetric LSH (ALSH) was proposed in~\cite{SL14} and several ALSH schemes for MIPS have been proposed since~\cite{SL14,SL15a,HMFFT18}. In our work, we use the ALSH construction of~\cite{SL15a} that reduces MIPS to Maximum Cosine Similarity Search (MCSS). We then use the cross-polytope LSH of~\cite{AILRS15} to solve MCSS efficiently.

%% file: prelims.tex
\section{Preliminaries}
\label{sec:prelims}

\subsection{LSH and nearest neighbor search}
\label{subsec:lsh-nns}
For a similarity measure $S$ on a domain $D$, define a ball of radius $s$ centered around point $q$ as $B(q, s) \coloneqq \{p : S(q, p) \geq s\}$. 
\begin{defn}[Locality sensitive hashing]
\label{defn:lsh}
A family of functions $\HF = \{h : D \rightarrow U\}$ is called $(c, s, p_1, p_2)$-sensitive for domain $D$ with similarity measure $S$ if for any $p, q \in D$:
\begin{itemize}
\item if $p \in B(q, s)$ then $\Pr_{h \in \HF}[h(q) = h(p)] \geq p_1$,  
\item if $p \notin B(q, cs)$ then $\Pr_{h \in \HF}[h(q) = h(p)] \leq p_2$.
\end{itemize}
\end{defn}

In order for an LSH family to be useful, it must satisfy inequalities $p_1 > p_2$ and $c < 1$. Note that locality sensitive hashing can be defined both in terms of a similarity measure or a distance measure. The similarity measure definition is convenient for our purposes in this paper.

An LSH family can be used to solve approximate near neighbor search.
\begin{defn}[Approximate near neighbor search]
\label{defn:apx-near}
The $(c, s)$-approx-imate near neighbor search problem for $c < 1$ with failure probability $f$ is to construct a data structure over a set of points $P$ in domain $D$ with similarity measure $S$ supporting the following query:
given any query point $q \in D$, if $p \in B(q, s)$ for some $p \in P$, then report some $p' \in P \cap B(q, cs)$, with probability $1 - f$.
\end{defn}

\begin{thm}[Restatement of Theorem 3.4 in~\cite{HIM12}]
\label{thm:lsh2nn}
Let $p_1, p_2 \in (0,1)$, $c < 1$ and $\rho = \log(1/p_1)/\log(1/p_2) < 1$. Given an $(c, s, p_1, p_2)$-sensitive family $\HF$ for a domain $D$, there exists a data structure for $(c, s)$-approximate near neighbor search over a set $P \subset D$ of at most $n$ points with constant failure probability requiring $O(n^\rho \log n)$ query time and $O(n^{1+\rho})$ space.
\end{thm}



The notion of asymmetric LSH was defined in~\cite{SL14}, who also showed approximate near neighbor search can also be solved using asymmetric LSH.
\begin{defn}[Asymmetric LSH]
A family of functions $\HF$, along with the two vector functions $P: R^d \rightarrow R^{d'}$ ({\em preprocessing transformation}) and $Q: R^d \rightarrow R^{d'}$ ({\em query transformation}), is called $(c, s, p_1, p_2)$-sensitive for $R^d$ with similarity measure $S$ if for any $p, q \in R^d$:
\begin{itemize}
\item if $S(q, p) \geq s$, then $\Pr_{h \in \HF}[h(Q(q))) = h(P(p))] \geq p_1$,
\item if $S(q, p) \leq cs$, then $\Pr_{h \in \HF}[h(Q(q)) = h(P(p))] \leq p_2$.
\end{itemize}
\end{defn}


\begin{defn}[Approximate maximum inner product search]
\label{defn:apx-mips}
Given a collection $P \subset R^d$ of size $n$, the $c$-approximate Maximum Inner Product Search (MIPS) with failure probability $f$ is to support the following query: given an input query point $q \in R^d$, find $p' \in P$ such that
\[ q^Tp' \geq c \max_{p \in P}q^T p,\]
with probability $1-f$.
\end{defn}

A connection between inner product and cosine similarity was established by Shrivastava and Li~\cite{SL15a} defining the following preprocessing and query transformations from $R^d$ to $R^{d+m}$ ($m = 3$ suffices):
\[P(x) \coloneq [x; 1/2 - ||x||_2^2; 1/2 - ||x||_2^4; \ldots; 1/2 - ||x||_2^{2^m}, \]
\[Q(x) \coloneqq [x; 0; 0; \ldots; 0]. \]

This lets one use an LSH scheme for cosine similarity to obtain an asymmetric LSH scheme for inner product, and thereby solve approximate MIPS. In this paper, we use the cross-polytope LSH of Andoni et al.~\cite{AILRS15} for cosine similarity along with the preprocessing and query transformations of~\cite{SL15a} for this purpose.


\subsection{Potential match search}
\label{subsec:prob}
Let $p_1, p_2, \ldots p_l$ be points on earth. Let $<p_1, p_2, \ldots, p_l>$ denote a route that starts at point $p_1$, traverses $p_2$, \ldots $p_{l-1}$ in order and ends at $p_l$, and $C(<p_1, p_2, \ldots, p_l>)$ denote its cost. Given a ride $r$, let $r_s$ and $r_t$ denote its pickup and dropoff locations respectively. The cost of the ride $r$, denoted by $C(r)$ abusing notation, is defined as $C(<r_s, r_t>)$. The cost can be measured in terms of travel distance, travel duration or any other metric linear in the route.

Fix a cost function. The utility of matching two rides comes from the cost savings achieved by serving them together compared to serving them individually. More formally, abusing notation, let $C(\{r, r'\})$ denote the cost of serving rides $r$ and $r'$ together: 
\begin{equation*}
    \begin{split}
        C(\{r, r'\}) = \min\{& C(<r_s, r'_s, r_t, r'_t>), C(<r_s, r'_s, r'_t, r_t>), \\
        & C(<r'_s, r_s, r'_t, r_t>), C(<r'_s, r_s, r_t, r'_t>)\}.
    \end{split}
\end{equation*}
We will assume that it is feasible to match two rides $r$ and $r'$ only if the matching satisfies a maximum allowable delay constraint for each of the rides. The feasibility function $F(\{r, r'\})$ is $1$ if it is feasible to match $r$ and $r'$, $0$ otherwise. The matching utility $U$ of matching rides $r$ and $r'$ together is then defined as:
\[ U(\{r, r'\}) = ( C(r) + C(r') - C(\{r, r'\})) * F(\{r, r'\}).\]

Note that in order to maximize the total matching utility from all the matches made, it may not be desirable to match a given ride with the one with which it has the highest matching utility. Therefore, the approximate potential match search is to find $k$ matches for every ride that approximately maximize matching utility, for a small enough $k$ so that a sparse, utility-aware shareability network can be constructed. An optimal (non-bipartite) matching on this sparse shareability network then yields approximately optimal total matching utility.

\begin{defn}[Approximate potential match search]
\label{defn:apx-pms}
Let $R$ be a set of rides. Let $S$ be the similarity measure between two rides defined as the matching utility between them. The $(c, s, k)$-approximate Potential Match Search (PMS) for $c < 1$ with failure probability $f$ is to construct a data structure over the set $R$ supporting the following query:
given any query ride $q$, if there exists $k$ rides $r_1, r_2, \ldots, r_k \in R$ such that $S(q, r) \geq s \ \forall r \in \{r_1, r_2, \ldots, r_k\}$, then report some $k$ rides $r'_1, r'_2, \ldots, r'_k \in R$ such that $S(q, r') \geq cs \ \forall r' \in \{r'_1, r'_2, \ldots, r'_k\}$, with probability $1 - f$.
\end{defn}

\subsection{Heuristic methods}
\label{subsec:heuristic}
To the best of our knowledge, search for approximate ride matches has not been formalized before this work and there are no principled methods known. The state-of-the-art heuristic methods include:
\begin{itemize}
    \item {\bf CLOSEBY:} $k$ nearest rides to the given ride with respect to pickup location are chosen as potential matches.
    
    \item {\bf HAVERSINE:} $k$ rides with highest Haversine matching utilities with the given ride are chosen as the potential matches. Haversine matching utility is defined with respect to Haversine cost function which is the Haversine distance between two points.
    
    \item {\bf CLOSEBY-HAVERSINE:} A hybrid approach of the above two where a sufficiently large number of rides nearest to the given ride with respect to pickup location are first selected, and the $k$ among them with highest Haversine matching utilities with the given ride are then chosen as potential matches.
\end{itemize}

%% file: algorithm.tex
\section{Algorithm for approximate PMS}
\label{sec:alg}

We first find a suitable representation for rides in a high-dimensional ambient space and define a similarity measure between two rides that approximately captures the matching utility between them. We then use the asymmetric transformations of~\cite{SL15a} (from MIPS to MCSS) combined with cross-polytope LSH construction of~\cite{AILRS15} (for MCSS) to obtain our algorithm for approximate PMS.

\subsection{Spatial ride match}
\label{subsec:spatial}
Consider two rides $r$ and $r'$ as shown in Figure~\ref{fig:rides} and suppose it is feasible to match them together. For brevity, we denote the pickup $r_s$ and dropoff $r_t$ for ride $r$ simply as $s$ and $t$ respectively. Similarly, the pickup and dropoff of ride $r'$ are denoted by $s'$ and $t'$ respectively. Let $<s, a, b, c, d, t>$ and $<s', a', b', c', t'>$ denote the routes of rides $r$ and $r'$ respectively.

\begin{figure}
  \includegraphics[width=\linewidth]{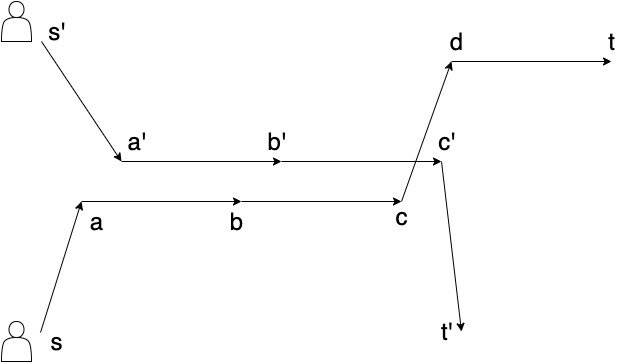}
  \caption{Two rides $r$ and $r'$. The route for $r$ is $<s, a, b, c, d, t>$ and for $r'$ is $<s', a', b', c', t'>$.}
  \label{fig:rides}
\end{figure}

Intuitively, since points $a, b$ and $c$ are close to points $a', b'$ and $c'$ respectively, the matching utility, which is the cost savings by matching $r$ and $r'$ together, is approximately the cost $C(<a, b>)$ of the route segment $(a, b)$ plus the cost $C(<b, c>)$ of the route segment $(b, c)$. In order to capture the notion of spatial proximity of points, we discretize space and represent each point of the route by the space discretization it falls in. We can use discretization using geohashes or S2 cells for this purpose. Let the space discretized node for points $a$ and $a'$ be $A$, and so on, as shown in Figure~\ref{fig:rides-discretized}. The space discretized route for rides $r$ and $r'$ are $<S, A, B, C, D, T>$ and $<S', A, B, C, T'>$ respectively.

\begin{figure}
  \includegraphics[width=\linewidth]{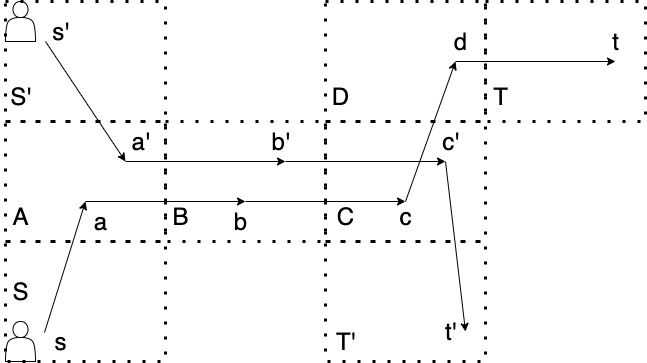}
  \caption{Two rides $r$ and $r'$. The space discretized route for $r$ is $<S, A, B, C, D, T>$ and for $r'$ is $<S', A, B, C, T'>$. The space discretization is shown in dotted lines and the space discretized nodes are shown in capital letters.}
  \label{fig:rides-discretized}
\end{figure}

Since we are interested in the cost of the overlapping segments, we represent rides by the set of space discretized edges in their routes instead of the sequence of space discretized nodes. We call this the spatial set representation of the ride. (Note that this is equivalent to 2-shingling viewing the route as a document; shingling is a popular technique in document similarity search.) Representing rides by its set of edges has the additional benefit that we get directionality for free. For example, two rides with exactly reversed routes of each other will have no edges in common.

Representing ride $r$ with the edge set $\{SA, AB, BC, CD, DT\}$ and ride $r'$ with $\{S'A, AB, BC, CT'\}$, the utility of matching them is roughly the sum of the costs of the edges in their intersection. Here, the cost of an edge between two space discretized nodes can be approximated as the cost between the mid-points of the discretizations. (For this approximation to work well, we must have sufficiently fine space discretization).

In view of the above discussion, ignoring match feasibility, approximate PMS can be solved by solving the approximate {\em Spatial Overlapping Match Search}:

\begin{defn}
\label{defn:apx-soms}
Let $R$ be a set of rides. Let $S$ be the similarity measure between two rides defined as the sum of the costs of the edges in the intersection of their spatial set representations. The $(c, s, k)$-approximate Spatial Overlapping Match Search (SOMS) for $c < 1$ with failure probability $f$ is to construct a data structure over the set $R$ supporting the following query:
given any query ride $q$, if there exists $k$ rides $r_1, r_2, \ldots, r_k \in R$ such that $S(q, r) \geq s \ \forall r \in \{r_1, r_2, \ldots, r_k\}$, then report some $k$ rides $r'_1, r'_2, \ldots, r'_k \in R$ such that $S(q, r') \geq cs \ \forall r' \in \{r'_1, r'_2, \ldots, r'_k\}$, with probability $1 - f$.
\end{defn}

For simplicity of exposition, let us assume that the cost of each edge is unit. This assumption is somewhat justified if it can be ensured that the routing engine returns routes with approximately equidistant adjecent nodes; however, we will relax it soon. With this assumption, the similarity between two rides is simply the cardinality of the intersection of their spatial set representations. Since cardinality of set intersection can be computed as the inner product of the corresponding characteristic vectors, we represent each ride by the characteristic vector of its spatial set representation. In other words, the ambient space is a high-dimensional space where each dimension is indexed by a space discretized edge. The spatial vector representation of a ride is a vector in this space with all $0$'s except a $1$ along a dimension iff the corresponding edge is in the spatial set representation of the ride. The cardinality of the intersection between the set representations of two rides is the inner product between their vector representations. Approximate SOMS can thus be solved by solving approximate MIPS.

\paragraph{Relaxing unit edge cost assumption.} Let us now relax the assumption of unit cost per edge. We will define two vector representations for each ride: {\em spatial preprocessing vector representation} used while creating the asymmetric LSH dataset, and {\em spatial query vector representation} used while querying the asymmetric LSH dataset to find overlapping matches for the ride. Both the vector representations are defined in the same ambient space as defined above.

\begin{itemize}
\item {\bf spatial preprocessing vector representation:} A vector with all $0$'s except for along dimensions for which the corresponding edge is in the set representation of the ride, in which case it is the cost of that edge.
\item {\bf spatial query vector representation:} A vector with all 0's except for along dimensions for which the corresponding edge is in the set representation of the ride, in which case it is $1$.
\end{itemize}

It is easy to see that with the above representations, the similarity between two rides $r$ and $q$ can be computed as the inner product between the spatial preprocessing vector representation of $r$ and the spatial query vector representation of $q$. Thus, like before, approximate SOMS can be solved via approximate MIPS.

\subsection{Spatio-temporal ride match}
\label{subsec:spatio-temporal}
Let us now turn our attention to accounting for match feasibility. Let $\Delta T$ be the maximum allowable delay for any ride. We discretize time into intervals of length $2\Delta T$. We define a space-time discretized node as a pair consisting of a space discretization and a time discretization. A space-time discretized edge is simply an ordered pair of such nodes. The space discretized nodes in the route of a ride can be annotated with the discretized time of reaching there if served without any delay. Thus, the route of a ride can be represented as a sequence of space-time discretized nodes, or equivalently, as a set of space-time discretized edges. We call this the spatio-temporal set representation of the ride.

If rides $r$ and $r'$ are feasible to match and an overlapping segment contributes to the matching utility between them, the segment must be traveled for the unmatched rides within time $\Delta T$ of each other. This implies there is good chance that the spatio-temporal set representations of both the rides share the corresponding space-time discretized edge. Thus, approximate PMS can be solved by solving the approximate {\em (Spatio-Temporal) Overlapping Match Search}:

\begin{defn}
\label{defn:apx-oms}
Let $R$ be a set of rides. Let $S$ be the similarity measure between two rides defined as the sum of the costs of the edges in the intersection of their spatio-temporal set representations. The $(c, s, k)$-approximate (Spatio-Temporal) Overlapping Match Search (OMS) for $c < 1$ with failure probability $f$ is to construct a data structure over the set $R$ supporting the following query:
given any query ride $q$, if there exists $k$ rides $r_1, r_2, \ldots, r_k \in R$ such that $S(q, r) \geq s \ \forall r \in \{r_1, r_2, \ldots, r_k\}$, then report some $k$ rides $r'_1, r'_2, \ldots, r'_k \in R$ such that $S(q, r') \geq cs \ \forall r' \in \{r'_1, r'_2, \ldots, r'_k\}$, with probability $1 - f$.
\end{defn}

Like before, we can construct two spatio-temporal vector representations for a ride. Each dimension of the ambient vector space is now indexed by a space-time discretized edge. The {\em spatio-temporal preprocessing vector representation} of a ride is a vector of all $0$'s in this space except for any dimension present in its spatio-temporal set representation where it is the cost of the edge. Similarly, the {\em spatio-temporal query vector representation} of a ride is a vector of all $0$'s except for any dimension present in its spatio-temporal set representation where it is $1$. It is now easy to see that we can solve approximate OMS via MIPS using these two vector representations.

\subsection{The algorithm}
\label{subsec:alg}
We are now ready to formally present our algorithm for approximate potential match search. As noted in the previous subsection, our strategy is to solve it via solving approximate OMS. The algorithm is formally presented in Algorithm~\ref{alg:apx-oms}.

\begin{algorithm}
\KwIn{A set $R$ of $n$ rides.}
\KwOut{For each ride $r$, a set $S_r$ with $|S_r| = k$.}
Let $U = 0.75$\;
\For{each ride $r$ in R}{
  Construct the spatio-temporal preprocessing vector representation $p_r$ of ride $r$\;
  Normalize the vector $p_r$ to have $||p_r||_2 \leq U$ (See~\cite{SL15a}\;
  Apply the preprocessing transformation of~\cite{SL15a} (see Section~\ref{sec:prelims}) on $p_r$ to get $P(p_r)$ with $m = 2$\;
}
Construct the LSH dataset using cross-polytope LSH of~\cite{AILRS15} with vectors $P(p_r)$, $\forall r \in R$\;
\For{each ride $r$ in R}{
  Construct the spatio-temporal query vector representation $q_r$ of ride $r$\;
  Normalize the vector $q_r$ to have $||q_r||_2 = 1$\;
  Apply the query transformation of ~\cite{SL15a} (see Section~\ref{sec:prelims}) on $q_r$ to get $Q(q_r)$ with $m = 2$\;
  Construct $S_r$ by retrieving $k$ nearest neighbors of $Q(q_r)$ from the LSH dataset\;
}
\caption{Algorithm for approximate PMS.}
\label{alg:apx-oms}
\end{algorithm}

We now show that our algorithm is efficient and has high probability of success.

\begin{thm}
\label{thm:apx-oms}
Given a set $R$ of $n$ rides, Algorithm~\ref{alg:apx-oms} solves $(c, s, k)$-approximate OMS requiring $O(n^\rho (k + \log n) \log k)$ time per query and $O(n^{1 + \rho} \log k)$ total space for some $\rho < 1$ depending on $c$ and $s$. Therefore, the total running time for the algorithm is $O(n^{1+ \rho} (k + \log n) \log k)$.
\end{thm}

\begin{proof}
Let $\HF$' be the family of cross-polytope LSH functions of~\cite{AILRS15}. Similar to arguments in~\cite{SL15a}, coupled with the preprocessing and the query transformations, they form a $(c, s, p_1, p_2)$-sensitive asymmetric LSH family $\HF$ for MIPS where $p_1, p_2$ depends on $c$ and $s$. Authors in~\cite{SL15a} show that $U = 0.75$ and $m = 2$ are reasonable parameter choices which we fixed in our algorithm.

We use the technique of amplification by concatenating $t = \log n/\log (1/p_2)$ randomly picked hash functions from $\HF$ to get the family of hash functions $\HG$. If the similarity between two rides $q$ and $r$, as defined in Definition~\ref{defn:apx-oms}, is at least $s$, we call them "similar". On the other hand, if the similarity between them is at most $cs$, we call them "dissimilar". We have:

\begin{itemize}
\item if $q$ and $r$ are similar, then $\Pr_{g \in \HG}[g(q) = g(r)] \geq p_1^t = n^{-\rho}$ where $\rho = \log (1/p_1)/\log (1/p_2) < 1$.
\item if $q$ and $r$ are dissimilar, then $\Pr_{g \in \HG}[g(q) = g(r)] \leq p_2^t = 1/n$
\end{itemize}

We construct the LSH data structure by using $L$ hash functions from $\HG$ to create $L$ hash tables. The LSH dataset is constructed by hashing the spatio-temporal preprocessing vector representations of the rides in $R$ to each of these $L$ hash tables. Fix a query ride $q$. For a "similar" ride $r$, probability of no hash collision in any of the $L$ hash tables is $(1 - n^{-\rho})^L$. We want this probability to be at most $f/k$. Setting $L = n^\rho \ln(k/f)$, we have:
\begin{align*} 
(1 - n^{-\rho})^L &\leq (e^{-n^{-\rho}})^L \\
&= (e^{-n^{-\rho}})^{n^\rho \ln(k/f)} \\
&= f/k
\end{align*}

By union bound, the probability that for any of the $k$ similar rides $r_1, r_2, \ldots, r_k$, there is no hash collision in any of the $L$ tables is upper bounded by $f$.
For a "dissimilar" ride $r$, the probability of hash collision in any of the $L$ hash tables is upper bounded by $L/n$.

Given a query ride, we hash the spatio-temporal query vector representation of $q$ using each of the $L$ hash functions to get the corresponding $L$ hash buckets. We retrieve all the rides from the those $L$ hash buckets. We check the similarity of each of these rides with $q$ and return $k$ non-dissimilar ones (i.e. similarity $> cs$ or "approximately similar"). 

Time needed for hashing the query is $O(Lt) = O(n^\rho \log k \log n)$. Time required for processing the first $k$ non-dissimilar rides is $O(kL) = O(k n^\rho \log k)$ as the same non-dissimilar ride can be retrieved at most $L$ times. Since the expected number of dissimilar rides is at most $n(L/n) = L$, time needed for processing the dissimilar rides is $O(L) = O(n^\rho \log k)$. Therefore, our algorithm solves $(c, s, k)$-approximate $k$-OMS requiring $O(n^\rho (k + \log n) \log k)$ time per query. The space requirement of the algorithm is $O(nL) = O(n^{1 + \rho} \log k)$.
\end{proof}

%% file: experiment.tex
\section{Experiments}
\label{sec:exp}

In order to validate our algorithm and the overall approach of solving approximate PMS via solving approximate OMS, we conducted extensive experimentation and evaluated the performance of our algorithm against benchmark approaches.

\subsection{Setup and methodology}
\label{subsec:setup_method}

\paragraph{Objective.} The problem we set out to experiment on is to search for approximate potential matches where the cost function is {\em travel duration}. Hence, the matching utility of two rides was defined as the savings in terms of travel duration achievable by matching and serving the two rides together.
 
\paragraph{Benchmarks.} We evaluated the performance of our algorithm (which we call LSH) against three benchmark state-of-the-art heuristics: CLOSEBY, HAVERSINE, and CLOSEBY-HAVERSINE. These heuristics are described in Section~\ref{sec:prelims}. Additionally, we computed the optimal matching utility achievable by any approach given unbounded time, and measured the performance of each approach against this optimal.

\paragraph{Metrics.} The metrics of interest to us are (i) total matching utility, (ii) computation time to construct the shareability network, and (iii) number of routing calls. Ideally, we want an algorithm that achieves high matching utility and is efficient both in terms of time for constructing edges of the shareability network and number of routing calls.

Routing calls are different from other computations done during the construction of the shareability network. This is because services often need to pay for these calls in terms of network latency as well as usage cost to a routing service provider. It is common for the routing service provider to allow batch routing calls. We accounted 10 ms (including computation time and network latency) for every batch routing call with at most 100 routing requests. This choice of accounting insulates our results against any time-varying irregularities in the execution efficiency of these calls. Additionally, we accounted for the usage cost by also tracking the number of routing calls separately. 

\paragraph{Experiments.} We wanted to evaluate our algorithm under different traffic patterns and varying load. For traffic pattern, we constructed two experiment scenarios:
\begin{itemize}
    \item {\bf Morning commute:} This scenario consisted of 21310 rides from NY yellow taxi trip data requested on 2016-06-08 (Wed) between 8:00 AM and 9:00 AM local time.
    
    \item {\bf Evening commute:} This scenario consisted of 19610 rides from NY yellow taxi trip data requested on 2016-06-08 between 6:00 and 7:00 PM local time. 
\end{itemize}
For varying load, we subsampled the rides in each scenario at subsampling rates of $20\%$, $40\%$, $60\%$, $80\%$ and $100\%$ (full dataset). Thus we ran a total of 10 experiments, 5 for each scenario.

To model feasibility of ride matching, we enforced a match constraint of maximum allowable pickup delay which rendered a match infeasible if either of the rides in the match incurred a delay of more than $10$ minutes due to the matching. We used travel duration between pickup locations as pickup delay.

\paragraph{Implementation.}  We implemented our algorithm as well as the three benchmark approaches to find top-k ($k = 10$) ride matches for each ride in the experiment.

The CLOSEBY approach was implemented using the Ball-Tree algorithm for near neighbor search using Haversine distance between pickup locations as the distance meaasure.

The HAVERSINE approach was implemented by computing pairwise Haversine matching utilities for every pair of rides, and then selecting the $k$ highest utility matches for each ride. For the CLOSEBY-HAVERSINE approach, we first found $1000$ matches for each ride using the CLOSEBY approach, and then selected $k$ matches among them with highest Haversine matching utility with the ride.

The LSH approach was implemented with geohash-$7$ space discretization and time discretization of $20$ minute intervals. The routes were obtained by running an instance of Open Source Routing Machine (OSRM) backend server locally. We set up the LSH data structure using the cross-polytope LSH family of~\cite{AILRS15} that is implemented in the FALCONN library. For this, we used $100$ hash tables and $14$ hash bits per table.

\paragraph{Evaluation.} The total matching utility achieved by an approach was computed by solving non-bipartite matching on the shareability network constructed. To recap, the nodes of the network were the rides in the experiment. A weighted (undirected) edge connecting two ride nodes was added to the network if either ride appeared as a match for the other. The edge was assigned a weight equal to the matching utility of the match between the two rides. The maximum weight matching obtained by solving this instance of non-bipartite matching gave us the total matching utility.

\subsection{Results}
\label{subsec:results}
Matching utilities obtained using different approaches under varying load for the morning commute scenario is presented in Figure~\ref{fig:morning_commute_utility}, and for the evening commute scenario in Figure~\ref{fig:evening_commute_utility}.

\begin{figure}[ht]
  \centering
  \includegraphics[width=\linewidth]{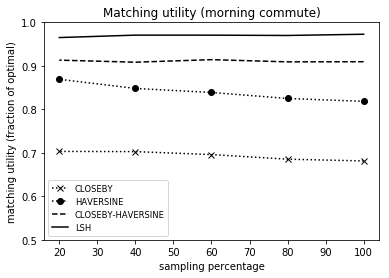}
  \caption{Matching utility as a fraction of the optimal for the morning commute scenario under varying load.}
  \label{fig:morning_commute_utility}
\end{figure}

\begin{figure}[ht]
  \centering
  \includegraphics[width=\linewidth]{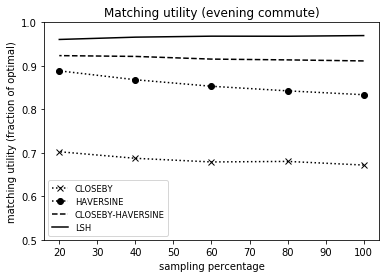}
  \caption{Matching utility as a fraction of the optimal for the evening commute scenario under varying load.}
  \label{fig:evening_commute_utility}
\end{figure}

Time to compute the shareability network using different approaches under varying load for the morning commute scenario is presented in Figure~\ref{fig:morning_commute_timing}, and for the evening commute scenario in Figure~\ref{fig:evening_commute_timing}.

\begin{figure}[ht]
  \centering
  \includegraphics[width=\linewidth]{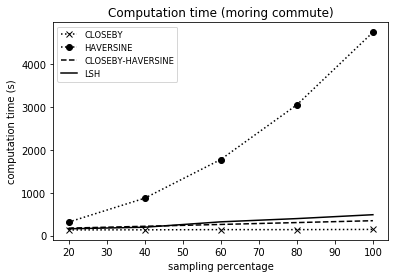}
  \caption{Shareability network computation time for the morning commute scenario under varying load.}
  \label{fig:morning_commute_timing}
\end{figure}

\begin{figure}[ht]
  \centering
  \includegraphics[width=\linewidth]{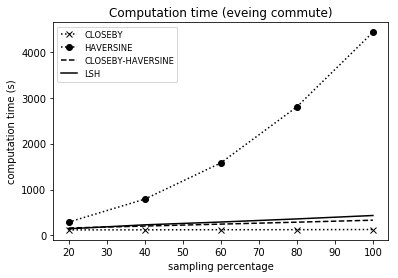}
  \caption{Shareability network computation time for the evening commute scenario under varying load.}
  \label{fig:evening_commute_timing}
\end{figure}

As for the number of routing calls, the LSH approach does require one upfront routing call per ride to get its route, while the other approaches only need to make routing calls after their match search process (i.e., once the edges of the shareability network are found and the edge weights need to be computed). However, it should be noted that the total number of routing calls is the same for all the approaches. To see this, note that computing the matching utility between two rides $r$ and $r'$ requires $8$ segment costs: $r_s$ to $r_t$, $r'_s$ to $r'_t$, $r_s$ to $r'_s$, $r'_s$ to $r_s$, $r_t$ to $r'_t$, $r'_t$ to $r_t$, $r_s$ to $r'_t$, and $r'_s$ to $r_t$. Thus for the entire network, we need $1$ routing call per node, and $6$ routing calls per edge. The calls per node are not dependent upon the outcome of the match search process and are known in advance. Hence those calls can be made upfront. These are precisely the routing calls the LSH approach requires for its match search process.

\subsection{Discussion}
\label{subsec:discussion}
In each of the 10 experiments, the optimal matching utility was found to be around $40\%$ of the total ride cost, suggesting huge potential for sharing rides and reducing cost. 

Out of the four approaches, CLOSEBY is the fastest as it simply needs to find near neighbors based on Haversine distance between pickup locations. Not surprisingly, it performs the worst in terms of matching utility. HAVERSINE is impractically slow, being required to perform an exhaustive search on a quadratic search space, even if without routing calls. Perhaps surprisingly, even with this exhaustive search, its performance in terms of matching utility is worse than that of CLOSEBY-HAVERSINE. This is because an exhaustive search finds matches that may have better overlap with a given ride but not feasible to match with it.

In each of the experiments, LSH consistently achieved about $6\%$ higher matching utility than CLOSEBY-HAVERSINE. This adds up to significant cost savings that can be shared by the riders, the drivers and the platform. It is important to note that CLOSEBY-HAVERSINE misses this opportunity not because of restricted computation time. In fact, it is faster than LSH, growing linearly with load while LSH grows slightly superlinearly. CLOSEBY-HAVERSINE suffers from its lack of knowledge of the underlying road network which forces it to miss out on achievable efficiency consistently. This can be confirmed by noting that the matching utility achieved by CLOSEBY-HAVERSINE stays almost constant at slightly over $90\%$ even with decreasing load to $20\%$. Since the number of CLOSEBY candidate matches stays constant at $1000$, this means CLOSEBY-HAVERSINE cannot achieve additional match utility even when it searches over a higher fraction of total rides.

We further note that computation time for LSH, although slightly larger than CLOSEBY-HAVERSINE for higher loads, is still practical.  The algorithm  is completely parallelizable, and the computation time for the shareability network can be brought down to about 30 secs with 10 parallel workers even for a gigantic ride pool of 20000 rides.  One can also employ dimensionality reduction techniques mentioned in the next section.

\begin{figure}[ht]
  \centering
  \includegraphics[width=\linewidth]{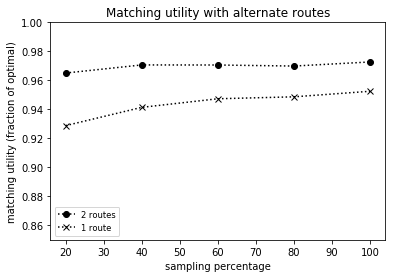}
  \caption{Matching utility as a fraction of the optimal with alternate routes.}
  \label{fig:utility_alt_routes}
\end{figure}

\paragraph{Alternate routes.} Having access to alternate routes for a ride helps the LSH approach to find better matches. A ride often has multiple routes with almost similar route costs. A feasible match may have a very good overlap with one of the routes that is not the minimum cost one. The computation time does increase with number of alternate routes, however. In our experiments, we used up to one additional alternate route whenever available (available for about $40\%$ of rides). Figure~\ref{fig:utility_alt_routes} compares LSH performance with one route and up to two routes.

\begin{figure}[ht]
  \centering
  \includegraphics[width=\linewidth]{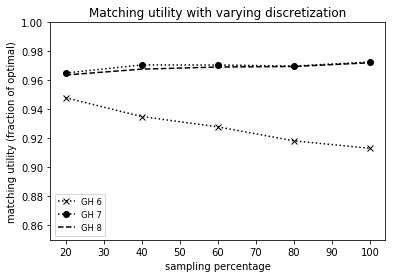}
  \caption{Matching utility as a fraction of the optimal with varying geohash discretization.}
  \label{fig:utility_geohash}
\end{figure}

\paragraph{Discretizations.} It is important to choose the space discretization for the LSH algorithm carefully. Choosing a discretization too coarse improves computation time as the dimension decreases. But it hurts matching utility by finding spurious route overlaps (segments too dissimilar look the same with space discretization) as well as missing some genuine ones (endpoints of an overlapping segment fuse with space discretization). This becomes more pronounced with increasing load as the chance of finding spurious overlaps increases. Choosing a discretization too fine increases the computation time significantly. In some cases, it might also hurt matching utility by causing the algorithm to miss some useful overlaps. Figure~\ref{fig:utility_geohash} compares performance with different space discretizations. In general, choosing a space discretization proportional to the typical length of a route edge works well (geohash-$7$ in our case).  For time discretization, it is best to choose a large interval. This does not increase the dimension and avoids missing feasible matches, but can lead to some infeasible matches which can later be filtered out.

%% file: extensions.tex
\section{Optimizations and Extensions}
\label{sec:opt_ext}
In this final section, we present some optimizations we learned while carrying out our experiments that we believe can provide significant performance boost in a real-world implementation. We also present several ideas for extensions that we believe can be exploited to extract additional value from our approach.

\subsection{Optimizations}
\label{subsec:opt}
As with machine learning algorithms, LSH-based algorithms can enjoy significant performance boost with the right tuning of parameters and domain-based insights.

\paragraph{Number of tables.} Increasing the number of tables in the LSH data structure improves success probability for finding nearest neighbors but also increases the space requirement and the query time. It is advisable to first choose a suitable value based on space availability and the dataset size.
    
\paragraph{Number of hash functions.} Increasing the number of hash functions per table causes fewer hash collisions improving query time but also decreases the success probability and necessitates larger number of tables. Fortunately, this can be handled using multi-probe query strategies described next. This parameter should be chosen roughly equal to the logarithm of the dataset size and further tuning should be done by performing parameter search jointly with number of query probes.
    
\paragraph{Multi-probe query.} It is possible to query a table multiple times using multi-probe query strategies to improve success probability without increasing the number of tables. The total number of probes across all table per query should be chosen jointly with number of hash functions per table. This can be achieved by iterating over the latter, and for each value doing a binary search over the former to achieve desired success probability and efficiency.
    
\paragraph{Dimension reduction.} The ambient vector space of ride representations has very large dimension and the data is very sparse along most dimensions. This is because most rides happen around busy areas while the remote areas only see few rides. This makes our approach ripe for employing dimensionality reduction techniques to transform the ride vectors to a low dimensional space before storing them in the LSH data structure. This can boost efficiency significantly without hurting success probability much.
    
\paragraph{Normalization.} Normalizing the dataset by centering the ride vectors along each dimension provides a significant gain in success probability.

\subsection{Extensions}
\label{subsec:ext}
Our approach can be extended in several interesting ways.

\paragraph{Driver matching.} Although we presented our approach and conducted our experiments with a match pool of only rides, our algorithm can be extended to also match drivers with rides. For this, each existing driver route (with possibly multiple passengers in the car) can be represented by its vector representation. The algorithm can then find matches between drivers and rides as well as between rides.

\paragraph{General linear cost function.} Our approach works for any cost function that is linear in the route traveled and utility defined as savings in terms of that cost function. Here, linearity means total cost of the route is the sum of costs incurred along the edges of the route. Examples of such cost functions include travel duration, travel distance, costs incurred along the route such as tolls and taxes, or any linear combination of such linear cost functions. This is useful as the cost of serving a ride is often estimated by a combination of time, distance and other factors.

\paragraph{Time varying cost function.} Our approach can naturally adapt to time-varying route costs. For example, travel durations vary significantly with time. As another example, a smart transportation system can impose time-varying tolls to regulate flow of traffic. As long as we have the accurate edge weights, our algorithm can find near-optimal matches for the current time, while the heuristic methods cannot adapt to these changes intelligently.

\paragraph{Incremental computation.} Our algorithm can be implemented to do incremental computation where the LSH data structure is persisted, and gets updated as new rides come (by adding them to the data structure) and rides get matched (by replacing them with the combined driver route). Dynamic LSH schemes are excellent candidates for such an implementation.

\paragraph{Hybrid rideshare.} A rideshare platform can have a ride pool where some rides are on-demand and real-time while others are scheduled ahead of time. Since our algorithm employs a spatio-temporal search, it can intelligently match rides from such hybrid pools respecting pickup constraints. For example, it can match a current ride with another scheduled to arrive 20 minutes from now, if the current ride is expected to reach near the pickup of the future ride in that time and they have good route overlap thereafter.
